\documentclass[12pt,reqno]{amsart}

\newcommand\version{September 18, 2021}

\usepackage{amsmath,amsfonts,amsthm,amssymb,amsxtra}
\usepackage{bbm} 

\newtheorem{theorem}{Theorem}

\newtheorem{lemma}[theorem]{Lemma}
\newtheorem{corollary}[theorem]{Corollary}

\theoremstyle{definition}

\theoremstyle{remark}

\newtheorem{remark}[theorem]{Remark}

\newcommand{\1}{\mathbbm{1}}

\renewcommand{\epsilon}{\varepsilon}

\renewcommand{\phi}{\varphi}
\newcommand{\R}{\mathbb{R}}

\DeclareMathOperator{\supp}{supp}
\DeclareMathOperator{\sgn}{sgn}

\begin{document}

\title[Minimizers for a one-dimensional interaction energy --- \version]{Minimizers for a\\ one-dimensional interaction energy}

\author{Rupert L. Frank}
\address[Rupert L. Frank]{Mathe\-matisches Institut, Ludwig-Maximilans Universit\"at M\"unchen, The\-resienstr.~39, 80333 M\"unchen, Germany, and Munich Center for Quantum Science and Technology, Schel\-ling\-str.~4, 80799 M\"unchen, Germany, and Mathematics 253-37, Caltech, Pasa\-de\-na, CA 91125, USA}
\email{r.frank@lmu.de}

\renewcommand{\thefootnote}{${}$} \footnotetext{\copyright\, 2021 by the author. This paper may be reproduced, in its entirety, for non-commercial purposes.\\
Partial support through U.S.~National Science Foundation grant DMS-1954995 and through the Deutsche Forschungsgemeinschaft (German Research Foundation) through Germany’s Excellence Strategy EXC-2111-390814868 is acknowledged.}

\begin{abstract}
	We solve explicitly a certain minimization problem for probability measures  in one dimension involving an interaction energy that arises in the modelling of aggregation phenomena. We show that in a certain regime minimizers are absolutely continuous with an unbounded density, thereby settling a question that was left open in previous works.
\end{abstract}

\maketitle

\section{Introduction and main results}

Motivated by applications in physics, mathematical biology and economics, a certain class of minimization problems involving a nonlocal interaction energy has attracted a lot of attention recently in the mathematics literature. In these models `particles' interact with each other through a pair potential that corresponds to a force that is repulsive on short distances and attractive on long ones. For background and also the connection to a class of time-dependent aggregation equations we refer to \cite{BaCaLaRa,BeTo,CaCaPa,CaFiPa,DaLiMC,DaLiMC2,Fr} and the references therein.

Here we study one very simple family of minimization problem of this type in one spatial dimension. This family has been study before, but a certain regime has been left open and it is our goal to complete this investigation. We denote by $P(\R)$ the set of Borel probability measure on $\R$ and for $\mu\in P(\R)$ and a parameter $\alpha>2$ we consider the energy functional
\begin{equation}
	\label{eq:energy}
	\mathcal E_\alpha[\mu] = \frac12 \iint_{\R\times\R} \left( \alpha^{-1}|x-y|^\alpha -2^{-1}|x-y|^2\right)d\mu(x)\,d\mu(y) \,.
\end{equation}
The corresponding minimization problem is
$$
E_\alpha:=\inf\left\{ \mathcal E_\alpha[\mu]  :\ \mu\in P(\R) \right\}.
$$
Recently, Davies, Lim and McCann \cite[Theorem 2.2]{DaLiMC} have shown that for $\alpha\geq 3$ the minimizers for $E_\alpha$ are precisely of the form $\mu = 2^{-1}(\delta_{a-1/2}+\delta_{a+1/2})$ for some $a\in\R$. Earlier, Kang, Kim, Lim and Seo \cite[Theorem 2]{KaKiLiSe} had shown that in the case $2<\alpha<3$, for any $m\in(0,1)$ and $a\in\R$ the measure $m\delta_{a-1/2}+(1-m)\delta_{a+1/2}$ is a saddle point for $E_\alpha$ (with respect to the $\infty$-Wasserstein metric) and therefore, in particular, not a minimizer. Finding the minimizer for $2<\alpha<3$ was explicitly stated as an open problem in \cite[Remark 1.6]{DaLiMC2}. As far as we know, up to now it was not even known whether or not optimizers are supported on a finite number of points, as the mild repulsivity assumption in \cite{BaCaLaRa,CaFiPa} barely fails in the above problem.

Our goal in this paper is to explicitly compute the minimial energy $E_\alpha$ and its minimizers. This settles the above open problem and shows, in particular, that for $2<\alpha<3$ minimizers are absolutely continuous and supported on an interval.

\begin{theorem}\label{main}
	Let $2<\alpha<3$ and set
	\begin{equation}
		\label{eq:ralpha}
		R_\alpha := \left( \frac{\sqrt\pi}{2}\ \frac{\Gamma(\frac{3-\alpha}2)}{\Gamma(\frac{4-\alpha}2)}\ \frac{\sin((\alpha-1)\frac\pi2)}{(\alpha-1)\frac\pi2} \right)^{\frac1{\alpha-2}}.
	\end{equation}
	Then
	\begin{equation}
		\label{eq:ealpha}
		E_\alpha = - \frac{\alpha-2}{2\alpha(4-\alpha)}\, R_\alpha^2 \,.
	\end{equation}
	Moreover, the infimum is attained if and only if for some $a\in\R$,
	$$
	d\mu(x) = C_\alpha^{-1}\, R_\alpha^{\alpha-2}\, (R_\alpha^2-(x-a)^2)^{-\frac{\alpha-1}2}\, \1(|x-a|<R_\alpha)\,dx \,,
	$$
	where $C_\alpha$ is an explicit normalization constant given in \eqref{eq:calpha}.
\end{theorem}

Denoting the measure in the theorem with $a=0$ by $\mu_\alpha$, it is not difficult to see that $\mu_\alpha\overset{\ast}\rightharpoonup 2^{-1}(\delta_{-1/2}+\delta_{+1/2})$ in $M(\R) = (C_0(\R))^*$ and $E_\alpha\to E_3$ as $\alpha\nearrow 3$. Thus one can think of the transition at $\alpha=3$ as a `singular bifurcation'.

\medskip

The same technique of proof used for Theorem \ref{main} allows us to solve the following related minimization problem. Let now $-1<\alpha<2$. For $\mu\in P(\R)$ we consider the energy functional
\begin{equation}
	\label{eq:energy2}
	\mathcal E_\alpha[\mu] = \frac12 \iint_{\R\times\R} \left( 2^{-1}|x-y|^2 - \alpha^{-1}|x-y|^\alpha \right)d\mu(x)\,d\mu(y)
\end{equation}
and the corresponding minimization problem
$$
E_\alpha:=\inf\left\{ \mathcal E_\alpha[\mu]  :\ \mu\in P(\R) \right\}.
$$
For $\alpha=0$ we understand $\alpha^{-1}|x-y|^\alpha$ as $\ln |x-y|$.

\begin{theorem}\label{main2}
	Let $-1<\alpha<2$ and define $R_\alpha$ by \eqref{eq:ralpha}. Then the conclusions of Theorem~\ref{main} remain true, except that the sign of the right side in \eqref{eq:ealpha} is switched.
\end{theorem}

As far as we know, Theorems \ref{main} and \ref{main2} are proved here for the first time. There are precursors in the literature, notably the works \cite{CaHu} by Carrillo and Huang and \cite{Agetal} by Agarwal et al., and we now discuss the similarities and differences with these works.

In  \cite{CaHu} it is shown that the measures appearing in Theorems \ref{main} and \ref{main2} satisfy `half' of the Euler--Lagrange relations corresponding to the minimization problem $E_\alpha$. This does not allow one to conclude that these measures are minimizers (and neither is this claimed in \cite{CaHu}). Let us be more precise concerning the Euler--Lagrange relations. These are well known for a large class of minimization problems including $E_\alpha$ and appear, for instance, in \cite{BaCaLaRa}. They consist of two parts, namely first, that the `potential' generated by a measure is constant on the support of this measure and second, that this potential is nowhere smaller than this constant. What is shown in \cite{CaHu} is the constancy on the support. (In fact, three different proofs of this fact are given by Polya and Szeg\H{o} in \cite[Hilfssatz I]{PoSz}.) We should stress, however, that \cite{CaHu} also has results about a larger class of minimization problems, which are outside the scope of this paper.

In \cite{Agetal} the statement of Theorem \ref{main2} appears, but from a mathematically rigorous perspective the argument given there is not completely satisfactory. More precisely, in \cite{Agetal} (a) the existence of a minimizing measure is taken for granted, (b) the minimizing measure is assumed to be absolutely continuous and supported on an interval, (c) rather precise properties of the Sonin inversion formula are used. Issue (a) can be overcome using relatively standard tools in the calculus of variations; see, e.g., \cite{CaCaPa,SiSlTo}. Issue (b) is quite subtle and we are not aware of general theorems from which one can deduce the desired properties. We do not doubt that the results concerning (c) are correct, but we would like to stress that the arguments take place in a rather singular setting with unbounded and barely integrable functions. Also, in absence of an easily accessible reference more selfcontained arguments might be preferable. In \cite{CaHu} (which is not quoted in \cite{Agetal}) the authors employed a similar approach via singular integral equations, but replaced some of the general theory by direct arguments.

In view of these previous works, our contribution in this paper is threefold. On the one hand, we provide a mathematically complete proof of Theorem \ref{main2} and, on the other hand, we show that a modification of these ideas can be used to prove Theorem~\ref{main}. Finally, we provide a proof without any direct analysis of singular integral equations.

Our proof of Theorem \ref{main2} is rather different from the arguments in \cite{CaHu,Agetal}. Namely, we rely on an elegant convexity argument that Lopes \cite{Lo} developed in the framework of a problem studied in \cite{BuChTo,FrLi}. This argument has proved useful in several other works since \cite{Lo} and has been slightly strengthened in \cite{CaDeDoFrHo,CaDeFrLe,DaLiMC,DaLiMC2} (extension to measures and characterization of cases of equality). The upshot of this argument is that one only needs to `guess' a solution to the Euler--Lagrange relations of the minimization problem and then this solution is automatically the unique (up to translations) minimizer. This conclusion is familiar from convex minimization problems and, indeed, Lopes's realization was that there is a `hidden' convexity. We present this argument in Lemma~\ref{lopes}. We emphasize that this argument also proves existence of a minimizer. Thus, it takes care of issue (a) mentioned above and makes (b) obsolete.

To guess a solution of the Euler--Lagrange relations we could follow the arguments in \cite{Agetal} based on singular integral equations and the Sonin inversion formula. Instead we opt for another approach, based on Fourier analysis. It relies on the computation of two Fourier transforms (namely \eqref{eq:fourier1} and \eqref{eq:fourier2}) that are probably not completely standard, but nevertheless contained in the usual tables. For the proof of Theorem \ref{main2} we use some analytic continuation arguments which are a bit lengthy, but not deep. Computationally, our approach is not more involved than that in \cite{Agetal}.

\medskip

We conclude this introduction with two remarks. First, it is interesting to compare the results in this paper with those for the minimization problem, depending on a parameter $\beta>-1$,
$$
\inf\left\{ \frac12 \iint_{[-1,1]\times [-1,1]} \beta^{-1} |x-y|^\beta\,d\mu(x)\,d\mu(y):\ \mu\in P([-1,1]) \right\}
$$
with a `strict confinement' to the interval $[-1,1]$. For this problem, the minimizing measure is absolutely continuous for $\beta<1$ (indeed, it is $Z_\beta^{-1}(1+x^2)^{-(1+\beta)/2}\,dx$) and equal to $(1/2)(\delta_{-1}+\delta_1)$ for $\beta\geq 1$. These results are classical; see, e.g., \cite[Section 7.5]{PoSz} for $\beta\geq 0$.

Our second remark concerns the question to which extent some structural properties of the minimizers in Theorems~\ref{main} and \ref{main2} are universal in the sense that they are valid in similar, but more general minimization problems. One question is which additional properties of interaction kernels vanishing like a negative quadratic at the origin guarantee that minimizing measures do not have atoms. (Recall that if the interaction kernel vanishes faster than quadratically, then minimizing measures are supported on a finite number of points \cite{CaFiPa}.) Moreover, all our minimizers are even and they are decreasing with respect to the distance from the center of symmetry for $\alpha<1$ and increasing for $\alpha>1$. It is natural to inquiry which structural assumptions on the interaction kernel ensure these properties.


\section{Proof of Theorem \ref{main}}

The following lemma is the theoretical backbone of our argument. It reduces the proof of our main result to finding a measure with certain properties. It is strongly influenced by Lopes's work \cite{Lo}.

\begin{lemma}\label{lopes}
	Let $2<\alpha<4$ and assume that there are $\mu\in P(\R)$ and $\eta\in\R$ such that
	$$
	\phi_\alpha(x) := \int_\R \left( \alpha^{-1} |x-y|^\alpha - 2^{-1} |x-y|^2\right)d\mu(y) \,,
	\qquad x\in\R \,,
	$$
	satisfies
	\begin{equation}
		\label{eq:lopesass}
		\phi_\alpha\geq \eta
	\quad\text{on}\ \R
	\qquad\text{and}\qquad
	\phi_\alpha = \eta
	\quad\text{on}\ \supp\mu \,.
	\end{equation}
	Then $\mu$ is the unique (up to translations) minimizer for $E_\alpha$ and $\eta=2E_\alpha$.
\end{lemma}

The proof shows that the second assumption in \eqref{eq:lopesass} can be slightly relaxed to requiring that $\phi_\alpha =\eta$ holds $\mu$-almost everywhere.

\begin{proof}
	Since the integrand in the definition of $\phi_\alpha$ is bounded from below, the integral is well-defined with values in $\R\cup\{+\infty\}$. Since $\phi_\alpha$ is finite on $\supp\mu$, we infer that $\int_\R |x|^\alpha\,d\mu<\infty$ and therefore $\phi_\alpha$ is finite everywhere and the center of mass of $\mu$ is well-defined. By translation invariance of the statement of Lemma \ref{lopes} we may assume that $\int_\R x \,d\mu(x) = 0$.
	
	Let $\tilde\mu\in P(\R)$. Our goal is to show that, if $\tilde\mu$ is not a translate of $\mu$, then $\mathcal E_\alpha[\tilde\mu]>\mathcal E_\alpha[\mu]$. We may assume that $\mathcal E_\alpha[\tilde\mu]<+\infty$ and, consequently, $\int_\R |x|^\alpha\,d\tilde\mu<\infty$ and the center of mass of $\tilde\mu$ is well-defined. By translation invariance of $\mathcal E_\alpha$ we may assume that $\int_\R x \,d\mu(x) = 0$. Our goal now is to prove that $\mathcal E_\alpha[\tilde\mu]>\mathcal E_\alpha[\mu]$ if $\tilde\mu\neq\mu$.
	
	For $\theta\in[0,1]$ we consider $\phi(\theta) := \mathcal E_\alpha[(1-\theta)\mu + \theta\tilde\mu]$ and show that (a) $\phi'(0)\geq 0$ and (b) $\phi''>0$ on $[0,1]$ if $\tilde\mu\neq\mu$. Since
	$$
	\phi(1) - \phi(0) = \int_0^1 \phi'(\theta)\,d\theta = \int_0^1 \left( \phi'(0) + \int_0^\theta \phi''(t)\,dt\right)d\theta = \phi'(0) + \int_0^1 (1-t) \phi''(t)\,dt \,,
	$$
	this implies that $\phi(1)>\phi(0)$ if $\tilde\mu\neq\mu$, which is the claimed strict inequality.
	
	We begin with the proof of (a). We write
	\begin{align*}
		\phi'(0) & = \iint_{\R\times\R} \left( \alpha^{-1} |x-y|^\alpha - 2^{-1} |x-y|^2\right) d\mu(x)\,d(\tilde\mu-\mu)(y) = \int_\R \phi_\alpha(y)\,d(\tilde\mu-\mu)(y) \\
		& = \int_\R \phi_\alpha(y)\,d\tilde\mu(y) - \int_\R \phi_\alpha(y)\,d\mu(y) \,.
	\end{align*} 
	The first and second assumptions in \eqref{eq:lopesass}, respectively, imply
	$$
	\int_\R \phi_\alpha(y)\,d\tilde\mu(y) \geq \eta \int_\R d\tilde\mu(y) = \eta
	\qquad\text{and}\qquad
	\int_\R \phi_\alpha(y)\,d\mu(y)  = \eta \int_\R d\mu(y) = \eta \,.
	$$
	Thus, $\phi'(0)\geq \eta - \eta =0$, as claimed.
	
	We finally turn to the proof of (b). Abbreviating $\nu:=\tilde\mu-\mu$ we have for all $\theta\in[0,1]$,
	\begin{align*}
		\phi''(\theta) & = \iint_{\R\times\R} \left( \alpha^{-1} |x-y|^\alpha - 2^{-1} |x-y|^2\right)d\nu(x)\,d\nu(y) \\
		& = \alpha^{-1} \iint_{\R\times\R} |x-y|^\alpha \,d\nu(x)\,d\nu(y) \,.
	\end{align*}
	In the last equality we expanded the square and used the fact that $\nu$ has vanishing integral and vanishing center of mass. The fact that $\phi''(\theta)\geq 0$ now follows from \cite[Theorem 2.4]{Lo}. Inspection of this proof (see also \cite[Corollary 3.2]{DaLiMC}) shows that one has, indeed, $\phi''(\theta)> 0$ if $\nu\neq 0$. This concludes the proof of the lemma.	
\end{proof}

\begin{lemma}\label{comp}
	Let $2<\alpha<3$. Then
	\begin{align}\label{eq:compint}
		& \int_{-1}^1 |x-y|^{\alpha} (1-y^2)^{-\frac{\alpha-1}2}\,dy \notag \\
		& = \alpha C_\alpha' x^2 + C_\alpha' +
		\begin{cases}
			0 & \text{if}\ |x|\leq 1 \,,\\
			\frac{\alpha(\alpha-1)(\alpha-2)}{2} C_\alpha \int_1^{|x|} (y^2-1)^{-\frac{3-\alpha}2} (|x|-y)^2\,dy & \text{if}\ |x|>1 \,,
		\end{cases}
	\end{align}
	with
	\begin{equation}
		\label{eq:calpha}
		C_\alpha := \sqrt\pi\ \frac{\Gamma(\frac{3-\alpha}2)}{\Gamma(\frac{4-\alpha}2)} 
		\qquad\text{and}\qquad
		C_\alpha' := \frac{\frac{(\alpha-1)\pi}2}{\sin\frac{(\alpha-1)\pi}2}\,.
	\end{equation}
\end{lemma}

\begin{proof}
	\emph{Step 1.} We begin by proving that
	\begin{equation}
		\label{eq:int}
		\int_{-1}^1 (\sgn(x-y)) |x-y|^{-3+\alpha} (1-y^2)^{-\frac{\alpha-1}2}\,dy =
		\begin{cases}
			0 & \text{if}\ |x|< 1 \,,\\
			C_\alpha (\sgn x)(x^2 -1)^{-\frac{3-\alpha}2} & \text{if}\ |x|>1 \,.
		\end{cases}
	\end{equation}
	We note that three proofs of this formula in the case $|x|<1$ appear in \cite[Hilfssatz I]{PoSz}. We argue differently, using Fourier transforms, and also derive the formula for $|x|>1$. According to \cite[(17.23.26), (17.34.10)]{GrRy} we have
	$$
	(\sgn x) |x|^{-3+\alpha} = - \frac{i\sin\frac{(\alpha-2)\pi}2\ \Gamma(\alpha-2)}{\pi} \int_\R (\sgn\xi) |\xi|^{-\alpha+2} e^{i\xi x}\,d\xi
	$$
	and
	\begin{equation}
		\label{eq:fourier1}
		(1-x^2)^{-\frac{\alpha-1}2}\1(|x|<1) = \frac{2^{-\frac{\alpha-2}2} \Gamma(\frac{3-\alpha}2)}{\sqrt\pi} \int_0^\infty \xi^{\frac{\alpha-2}2} J_{-\frac{\alpha-2}2}(\xi)\cos(\xi x) \,d\xi \,.
	\end{equation}
	Thus,
	\begin{align*}
		& \int_{-1}^1 (\sgn(x-y)) |x-y|^{-3+\alpha} (1-y^2)^{-\frac{\alpha-1}2}\,dy \\
		& = \frac{\sin\frac{(\alpha-2)\pi}2\ \Gamma(\alpha-2)\
		2^{\frac{4-\alpha}2}\ \Gamma(\frac{3-\alpha}2)}{\sqrt\pi}
		\int_0^\infty \xi^{-\frac{\alpha-2}2} J_{-\frac{\alpha-2}2}(\xi)\sin(\xi x)\,d\xi \,.
	\end{align*}
	Finally, according to \cite[(6.699.5)]{GrRy}
	\begin{equation}
		\label{eq:fourier2}
		\int_0^\infty \xi^{-\frac{\alpha-2}2} J_{-\frac{\alpha-2}2}(\xi)\sin(\xi x)\,d\xi
		= \begin{cases}
			0 & \text{if}\ 0<x<1 \,,\\
			\frac{\sqrt\pi 2^{-\frac{\alpha-2}2}}{\Gamma(\frac{\alpha-1}2)} (x^2-1)^{-\frac{3-\alpha}2}
			& \text{if}\ x>1 \,.
		\end{cases}
	\end{equation}
	This proves the claimed formula \eqref{eq:int} with the constant
	$$
	C_\alpha = \frac{2^{3-\alpha} \sin\frac{(\alpha-2)\pi}2\ \Gamma(\alpha-2)\ \Gamma(\frac{3-\alpha}2)}{\Gamma(\frac{\alpha-1}2)} \,.
	$$
	By Legendre's duplication formula and Euler's reflection formula, respectively,
	$$
	\frac{\Gamma(\alpha-2)}{\Gamma(\frac{\alpha-1}2)} = \frac{\Gamma(\frac{\alpha-2}2)}{2^{3-\alpha}\sqrt\pi}
	\quad\text{and}\quad
	\Gamma(\tfrac{\alpha-2}2) = \frac{\pi}{\sin\frac{(\alpha-2)\pi}2\ \Gamma(\frac{4-\alpha}2)} \,.
	$$
	Using these formulas we can bring $C_\alpha$ into the claimed form.
	
	\medskip
	
	\emph{Step 2.} We now show that the formula in the lemma follows from \eqref{eq:int} by triple integration. Indeed, integrating \eqref{eq:int} with respect to $x$ yields
	\begin{align}
		\label{eq:int1}
		& \int_{-1}^1 |x-y|^{\alpha-2} (1-y^2)^{-\frac{\alpha-1}2}\,dy \notag \\
		& = c_\alpha + 
		\begin{cases}
			0 & \text{if}\ |x|< 1 \,,\\
			(\alpha-2) C_\alpha \int_1^{|x|} (y^2 -1)^{-\frac{3-\alpha}2}\,dy & \text{if}\ |x|>1 \,,
		\end{cases}
	\end{align}
	with
	\begin{align*}
		c_\alpha & := \int_{-1}^1 |y|^{\alpha-2} (1-y^2)^{-\frac{\alpha-1}2}\,dy = \int_0^1 t^{(\alpha-3)/2} (1-t)^{-\frac{\alpha-1}2}\,dt = \Gamma(\tfrac{\alpha-1}2)\,\Gamma(\tfrac{3-\alpha}2) \\
		& = \frac{\pi}{\sin\frac{(\alpha-1)\pi}2} \,.
	\end{align*}
	Here we changed variables $y^2 =t$, expressed the beta function in terms of gamma functions and used Euler's reflection formula for the gamma function. Integration of \eqref{eq:int1} with respect to $x$ shows that
	\begin{align}
		\label{eq:int2}
		& \int_{-1}^1 (\sgn(x-y)) |x-y|^{\alpha-1} (1-y^2)^{-\frac{\alpha-1}2}\,dy \notag \\
		& = (\alpha-1) c_\alpha x +
		\begin{cases}
			0 & \text{if}\ |x|< 1 \,,\\
			(\alpha-1)(\alpha-2) C_\alpha (\sgn x) \int_1^{|x|} \int_1^{|y|} (z^2 -1)^{-\frac{3-\alpha}2}\,dz\,dy & \text{if}\ |x|>1 \,.
		\end{cases}
	\end{align}
	No additional integration constant appears since the left side is an odd function of $x$. The double integral on the right side of \eqref{eq:int2} equals
	$$
	\int_1^{|x|} \int_1^{|y|} (z^2 -1)^{-\frac{3-\alpha}2}\,dz\,dy = \int_1^{|x|} (z^2-1)^{-\frac{3-\alpha}2} (|x|-z)\,dz \,.
	$$
	One final integration with respect to $x$ shows that
	\begin{align}
		\label{eq:int3}
		& \int_{-1}^1 |x-y|^{\alpha} (1-y^2)^{-\frac{\alpha-1}2}\,dy = \frac{\alpha(\alpha-1)}2 c_\alpha x^2 + C_\alpha' \notag \\
		& +
		\begin{cases}
			0 & \text{if}\ |x|< 1 \,,\\
			\alpha(\alpha-1)(\alpha-2) C_\alpha \int_1^{|x|} \int_1^{|y|} (z^2-1)^{-\frac{3-\alpha}2} (|y|-z)\,dz\,dy & \text{if}\ |x|>1 \,,
		\end{cases}
	\end{align}
	where
	\begin{align*}
		C_\alpha' & = \int_{-1}^1 |y|^\alpha (1-y^2)^{-\frac{\alpha-1}2}\,dy = \int_0^1 t^{\frac{\alpha-1}2} (1-t)^{-\frac{\alpha-1}2}\,dt 
		= \Gamma(\tfrac{\alpha+1}2)\,\Gamma(\tfrac{3-\alpha}2) \\
		& = \tfrac{\alpha-1}2\,  \Gamma(\tfrac{\alpha-1}2)\,\Gamma(\tfrac{3-\alpha}2)
		= \frac{\frac{(\alpha-1)\pi}2}{\sin\frac{(\alpha-1)\pi}2}\,.
	\end{align*}
	The double integral on the right side of \eqref{eq:int3} equals
	$$
	\int_1^{|x|} \int_1^{|y|} (z^2-1)^{-\frac{3-\alpha}2} (|y|-z)\,dz\,dy = \frac12 \int_1^{|x|} (z^2-1)^{-\frac{3-\alpha}2} (|y|-z)^2\,dz \,.
	$$
	This completes the proof of \eqref{eq:compint}.
\end{proof}

\begin{corollary}\label{cor}
	Let $2<\alpha<3$ and let $R_\alpha$ be defined by \eqref{eq:ralpha}. Then the measure
	$$
	d\mu(x) = C_\alpha^{-1} R_\alpha^{\alpha-2} (R_\alpha^2-x^2)^{-\frac{\alpha-1}2} \1(|x|<R_\alpha)\,dx
	$$
	satisfies the assumptions of Lemma \ref{lopes} with
	$$
	\eta =  - \frac{\alpha-2}{\alpha(4-\alpha)}\, R_\alpha^2 \,.
	$$
\end{corollary}

\begin{proof}
	By Lemma \ref{comp} and scaling one has, for any $R>0$,
	\begin{align*}
		\frac1\alpha \int_{-R}^R |x-y|^\alpha \left( R^2-y^2 \right)^{-\frac{\alpha-1}2}dy = C_\alpha' x^2 + \alpha^{-1} C_\alpha' R^2 + R^2 f(x/R)
	\end{align*}
	with
	\begin{align*}
		f(x) & :=
		\begin{cases}
			0 & \text{if}\ |x|\leq 1 \,,\\
			\frac{(\alpha-1)(\alpha-2)}{2} C_\alpha \int_1^{|x|} (y^2-1)^{-\frac{3-\alpha}2} (|x|-y)^2\,dy & \text{if}\ |x|>1 \,.
		\end{cases}
	\end{align*}
	Moreover,
	\begin{align*}
		\frac12 \int_{-R}^R |x-y|^2 \left( R^2-y^2 \right)^{-\frac{\alpha-1}2}dy = \frac12 C_\alpha R^{-\alpha+2} x^2 + \frac12 \tilde C_\alpha R^{4-\alpha}
	\end{align*}
	with
	\begin{align*}
		\tilde C_\alpha = \int_{-1}^1 y^2 (1-y^2)^{-(\alpha-)/2}\,dy = \int_0^1 \sqrt t (1-t)^{-\frac{\alpha-1}2}\,dt = \frac{\Gamma(\frac32)\Gamma(\frac{3-\alpha}2)}{\Gamma(\frac{6-\alpha}2)} = \frac{\sqrt\pi\ \Gamma(\frac{3-\alpha}2)}{2\,\Gamma(\frac{6-\alpha}2)}
	\end{align*}
	and where we used the fact that, by a similar computation,
	$$
	\int_{-1}^1 (1-y^2)^{-\frac{\alpha-1}2}\,dy = C_\alpha \,.
	$$
	Choosing $R=R_\alpha$ and noting that $R_\alpha^{\alpha-2}= C_\alpha/ (2\,C_\alpha')$, we see that the coefficients of $x^2$ coincide and we obtain
	\begin{align*}
		& \int_{-R_\alpha}^{R_\alpha} \left( \alpha^{-1}|x-y|^\alpha - 2^{-1} |x-y|^2\right)\left( R_\alpha^2 -y^2 \right)^{-\frac{\alpha-1}2}dy \\
		& = - \left( 2^{-1} \tilde C_\alpha R_\alpha^{4-\alpha} - \alpha^{-1} C_\alpha' R_\alpha^2 \right) + R_\alpha^2 f(x/R_\alpha^2) \,.
	\end{align*}
	Since $\int_{-R_\alpha}^{R_\alpha} (R_\alpha^2 - y^2)^{-\frac{\alpha-1}2}\,dy = C_\alpha R_\alpha^{-\alpha+2}$, we see that $\mu\in P(\R)$. Since $f\geq 0$ with equality for $|x|\leq 1$ we see that $\mu$ satisfies the assumptions of Lemma \ref{lopes}. The constant appearing there is
	$$
	\eta = - \left( 2^{-1} \tilde C_\alpha R_\alpha^{4-\alpha} - \alpha^{-1} C_\alpha' R_\alpha^2 \right) C_\alpha^{-1} R_\alpha^{\alpha-2} = - \left( 2^{-1} \tilde C_\alpha - \alpha^{-1} C_\alpha' R_\alpha^{\alpha-2} \right) C_\alpha^{-1} R_\alpha^2 \,.
	$$
	Inserting first the definition of $R_\alpha$ and then the explicit form of $C_\alpha$ and $\tilde C_\alpha$ gives
	\begin{align*}
		& \left( 2^{-1} \tilde C_\alpha - \alpha^{-1} C_\alpha' R_\alpha^{\alpha-2} \right) C_\alpha^{-1}
		= 2^{-1} \left(  \tilde C_\alpha C_\alpha^{-1} - \alpha^{-1} \right) = \frac12 \left( (4-\alpha)^{-1} - \alpha^{-1} \right) \\
		& = \frac{\alpha - 2}{\alpha(4-\alpha)} \,.
	\end{align*}
	This completes the proof.	
\end{proof}

Of course, Theorem \ref{main} is an immediate consequence of Lemma \ref{lopes} and Corollary \ref{cor}.


\section{Proof of Theorem \ref{main2}}

Since the proof of Theorem \ref{main2} is similar to that of Theorem \ref{main}, we mostly focus on the differences. The analogue of Lemma \ref{lopes} reads as follows.

\begin{lemma}\label{riesz}
	Let $-1<\alpha<2$ and assume that there are $\mu\in P(\R)$ and $\eta\in\R$ such that
	$$
	\phi_\alpha(x) := \int_\R \left( 2^{-1} |x-y|^2 - \alpha^{-1} |x-y|^\alpha \right)d\mu(y) \,,
	\qquad x\in\R \,,
	$$
	satisfies
	\begin{equation*}
		\phi_\alpha\geq \eta
		\quad\text{on}\ \R
		\qquad\text{and}\qquad
		\phi_\alpha = \eta
		\quad\text{on}\ \supp\mu \,.
	\end{equation*}
	Then $\mu$ is the unique (up to translations) minimizer for $E_\alpha$ and $\eta=2E_\alpha$.
\end{lemma}

\begin{proof}
	The proof is rather similar to that of Lemma \ref{lopes}, except that the argument that $\phi''>0$ is more standard. Indeed, in the notation of the previous proof, we find
	$$
	\phi''(\theta) = -\frac{1}{\alpha} \iint_{\R\times\R} |x-y|^\alpha\,d\nu(x)\,d\nu(y) \,.
	$$
	For $-1<\alpha<0$ we use the fact that the Fourier transform of $|x-y|^\alpha$ is positive definite. For $0\leq\alpha<2$ (recall that we interpret $\alpha^{-1}|x-y|^\alpha$ as $\ln|x-y|$ for $\alpha=0$) we use the fact that the Fourier transform of $-|x-y|^\alpha$ is positive definite when restricted to signed measures with vanishing integral. This allows one to conclude the proof as before.
\end{proof}

\begin{lemma}\label{comp2}
	Formula \eqref{eq:compint} holds for $1<\alpha\leq 2$. Moreover, for $-1<\alpha< 2$ we have
	\begin{align*}
		& \int_{-1}^1 |x-y|^{\alpha} (1-y^2)^{-\frac{\alpha-1}2}\,dy = \alpha C_\alpha' x^2 + C_\alpha' \\
		& -
		\begin{cases}
			0 & \text{if}\ |x|\leq 1 \,,\\
			\alpha D_\alpha (|x|-1)^2 + \frac{\alpha(\alpha-1)(\alpha-2)}{2} C_\alpha \int_1^{|x|} \int_1^{|y|} \int_{|z|}^\infty (w^2 -1)^{-\frac{3-\alpha}2}\,dw\,dz\,dy & \text{if}\ |x|>1 \,,
		\end{cases}
	\end{align*}
with
$$
D_\alpha := C_\alpha\, \frac{\Gamma(\frac{\alpha+1}2)\,\Gamma(\frac{4-\alpha}2)}{\sqrt\pi} \,.
$$
\end{lemma}

\begin{proof}
	The first assertion follows easily by analytic continuation, since for fixed $x\in\R$ both sides of \eqref{eq:compint} are analytic in $\alpha$ in an open set in the complex plane containing $\{1<\alpha<3\}$. The restriction here to $\alpha>1$ comes from the integral on the right side and its converges near $y=1$. To prove the second assertion in the lemma we will construct an analytic continuation of that integral. To do so, we review the second step of the proof of Lemma \ref{comp}. The same analytic continuation argument shows that \eqref{eq:int1} holds for $1<\alpha<3$. Restricting ourselves to $1<\alpha<2$, we can rewrite \eqref{eq:int1} as
	\begin{align}
		\label{eq:int1a}
		& \int_{-1}^1 |x-y|^{\alpha-2} (1-y^2)^{-\frac{\alpha-1}2}\,dy \notag \\
		& = c_\alpha +
		\begin{cases}
			0 & \text{if}\ |x|< 1 \,,\\
			c_\alpha^{(1)} - (\alpha-2)\, C_\alpha\, \int_{|x|}^\infty (y^2 -1)^{-\frac{3-\alpha}2}\,dy & \text{if}\ |x|>1 \,,
		\end{cases}
	\end{align}
	with
	\begin{align*}
		c_\alpha^{(1)} & := (\alpha-2) C_\alpha \int_1^\infty (y^2 -1)^{-\frac{3-\alpha}2}\,dy = \frac{(\alpha-2)C_\alpha}2 \int_0^1 (1-s)^{-\frac{3-\alpha}2} s^{-\alpha/2}\,ds \\ 
		& = \frac{(\alpha-2) C_\alpha}{2}\, \frac{\Gamma(\frac{\alpha-1}2)\Gamma((2-\alpha)/2)}{\sqrt{\pi}} = - C_\alpha \, \frac{\Gamma(\frac{\alpha-1}2)\Gamma(\frac{4-\alpha}2)}{\sqrt{\pi}} \,.
	\end{align*}
	We integrate \eqref{eq:int1a} with respect to $x$ and obtain
	\begin{align}
		\label{eq:int2a}
		& \int_{-1}^1 (\sgn(x-y)) |x-y|^{\alpha-1} (1-y^2)^{-\frac{\alpha-1}2}\,dy = (\alpha-1) c_\alpha x \notag \\
		& +
		\begin{cases}
			0 & \text{if}\ |x|< 1 \,,\\
			(\alpha-1)\, c_\alpha^{(1)}\, (\sgn x)\,(|x|-1) & \\
			\qquad - (\alpha-1)(\alpha-2)\, C_\alpha\, (\sgn x)\, \int_1^{|x|} \int_{|y|}^\infty (z^2 -1)^{-\frac{3-\alpha}2}\,dz\,dy & \text{if}\ |x|>1 \,.
		\end{cases}
	\end{align}
	The important observation now is that
	$$
	(\alpha-1)\, c_\alpha^{(1)} = - 2\, C_\alpha \, \frac{\Gamma(\frac{\alpha+1}2)\Gamma(\frac{4-\alpha}2)}{\sqrt{\pi}}
	$$
	is analytic in complex open set containing $-1<\alpha<4$. Thus, formula \eqref{eq:int2a} holds at least for $0<\alpha<2$. (We restrict ourselves here to $\alpha>0$ so that the integral on the left side converges absolutely.) We also note that $\int_{|y|}^\infty (z^2-1)^{-\frac{3-\alpha}2}\,dz$ behaves like a constant times $(|y|-1)^{-(1-\alpha)/2}$ as $|y|\to 1$ and therefore it is integrable near $|y|=1$ as long as $\alpha>-1$.
	
	Integrating \eqref{eq:int2a} with respect to $x$ we obtain
		\begin{align}
		\label{eq:int3a}
		& \int_{-1}^1 |x-y|^{\alpha} (1-y^2)^{-\frac{\alpha-1}2}\,dy = \frac{\alpha(\alpha-1)}2 c_\alpha x^2 + C_\alpha' \notag \\
		& =
		\begin{cases}
			0 & \text{if}\ |x|< 1 \,,\\
			\frac{\alpha (\alpha-1)}2\, c_\alpha^{(1)}\, (|x|-1)^2 & \\
			\qquad - \alpha (\alpha-1)(\alpha-2)\, C_\alpha\, \int_1^{|x|} \int_1^{|y|} \int_{|z|}^\infty (w^2 -1)^{-\frac{3-\alpha}2}\,dw\,dz\,dy & \text{if}\ |x|>1 \,.
		\end{cases}
	\end{align}
	This formula, which we derived under the assumption $0<\alpha<2$ extends, by analytic continuation to $-1<\alpha<2$. This completes the proof of the lemma.	
\end{proof}

\begin{remark}
	There is an partially alternate proof of Lemma \ref{comp2}, which proceeds by verifying the claimed formulas using Fourier transforms in the spirit of our proof of Lemma \ref{comp}. More precisely, one verifies \eqref{eq:int3a} for $-1<\alpha<0$, \eqref{eq:int2a} for $0<\alpha<1$ and \eqref{eq:int1a} for $1<\alpha<2$. (In these cases the Fourier transform of the convolution kernels is welldefined without the need of analytic continuation.) The relevant formulas are \cite[(6.699.1) and (6.699.2)]{GrRy}. The disadvantage of such a proof is that the `remainder terms' are expressed as hypergeometric functions and one needs some of their properties. For this reason we chose the above somewhat lengthy, but elementary proof.
\end{remark}

\begin{corollary}
	Let $-1<\alpha<2$ and define $R_\alpha$ and $\mu$ as in Corollary \ref{cor}. Then $\mu$ satisfies the assumptions of Lemma \ref{riesz} with
	$$
	\eta =  - \frac{2-\alpha}{\alpha(4-\alpha)}\, R_\alpha^2 \,.
	$$
\end{corollary}

\begin{proof}
	For $1<\alpha<2$ we argue in exactly the same way as in the proof of Corollary~\ref{cor}. Concerning the sign of the remainder term we note that there is change of sign in the definition of $\phi_\alpha$ when $\alpha$ passes through $2$, but this change is compensated by the factor $\alpha-2$ in $f$. Thus everything goes through as before, except that the change of sign of $\phi_\alpha$ leads to a change of sign of $\eta$.
	
	In the case $-1<\alpha\leq 1$ we have by Lemma \ref{comp2}, for any $R>0$,
	\begin{align*}
		\frac1\alpha \int_{-R}^R |x-y|^\alpha \left( R^2-y^2 \right)^{-\frac{\alpha-1}2}dy = C_\alpha' x^2 + \alpha^{-1} C_\alpha' R^2 - R^2 g(x/R)
	\end{align*}
	with
	\begin{align*}
		g(x) & :=
		\begin{cases}
			0 & \text{if}\ |x|\leq 1 \,,\\
			D_\alpha (|x|-1)^2 + \frac{(\alpha-1)(\alpha-2)}{2} C_\alpha \int_1^{|x|} \int_1^{|y|} \int_{|z|}^\infty (w^2 -1)^{-\frac{3-\alpha}2}\,dw\,dz\,dy & \text{if}\ |x|>1 \,.
		\end{cases}
	\end{align*}
	The assertion in the corollary follows from the fact that both terms in the definition of $g(x)$ for $|x|>1$ are nonnegative. The rest follows from computations that are similar as in the proof of Corollary \ref{comp} and that are omitted.
\end{proof}


\bibliographystyle{amsalpha}

\end{document}